\documentclass[journal]{IEEEtran}

\ifCLASSINFOpdf
\else
   \usepackage[dvips]{graphicx}
\fi
\usepackage{url}

\usepackage{graphicx}

\usepackage{bm}
\usepackage{cite}
\usepackage{xcolor}
\usepackage{pifont}
\usepackage{amsthm}
\usepackage{comment}
\usepackage{amsmath}
\usepackage{amssymb}
\usepackage{balance}
\usepackage{amsfonts}
\usepackage{arydshln}
\usepackage{graphicx}
\usepackage{setspace}
\usepackage{balance}
\usepackage{algorithm}
\usepackage{algorithmic}
\usepackage{caption}

\newcounter{MYtempeqncnt}
\newtheorem{lemma}{Lemma}
\newtheorem*{lemma*}{Lemma}

\newtheorem{theorem}{Theorem}
\newtheorem*{theorem*}{Theorem}

\newtheorem{corollary}{Corollary}

\begin{document}

\title{Outage Probability of Intelligent Reflecting Surfaces Assisted Full Duplex Two-way Communications}

\author{Binyu~Lu,~\IEEEmembership{Student~Member,~IEEE,}
	Rui~Wang,~\IEEEmembership{Senior Member,~IEEE,}
	Yiming~Liu,~\IEEEmembership{Student~Member,~IEEE}
\thanks{B. Lu, R. Wang and Y. Liu are with the College
	of Electrical and Information Engineering, Tongji University, Shanghai, 201804, 
	China, E-mail: tjabinl@163.com, ruiwang@tongji.edu.cn, ymliu\_970131@tongji.edu.cn }
}

\maketitle

\begin{abstract}
In this letter, we study the outage probability of intelligent reflecting surface (IRS) assisted full duplex two-way communication systems, which characterizes the performance of overcoming the transmitted data loss caused by long deep fades. To this end, we first derive the probability distribution of the cascaded end-to-end equivalent channel with an arbitrarily given IRS beamformer. Our analysis shows that deriving such probability distribution in the considered case is more challenging than the case with the phase-matched IRS beamformer. 
Then, with the derived probability distribution of the equivalent channel, we obtain the closed-form expression of the outage probability performance. It theoretically shows that the reflecting element number has a conspicuous effect on the improvement of the system reliability. 
Extensive numerical results verify the correctness of the derived results and confirm the superiority of the considered IRS assisted two-way communication system comparing to the one-way counterpart. 
\end{abstract}


\IEEEpeerreviewmaketitle

\section{Introduction}

\IEEEPARstart{I}{ntelligent} reflecting surface (IRS)
has recently emerged as a promising technology to reconstruct high-quality channel links in the sixth-generation wireless networks \cite{wu2019towards}.
On the one hand, based on the massive reflecting elements, whose phase shifts can be adjusted in a desired direction independently, IRS can flexibly control the propagation environment of radio signals \cite{di2019smart}.
On the other hand, full duplex two-way communication systems, in which users can transmit and receive messages simultaneously over the same channel, can significantly improve spectral efficiency \cite{sabharwal2014band}. 
Since no transmit power consumption of IRS, the case of IRS includes nonexistence of rate loss in full-duplex relaying \cite{9025235}. Therefore, introducing IRS to two-way communications can enhance reliability and availability of wireless networks.

The system performance of IRS-assisted communications has been investigated in a lot of literature in terms of energy efficiency\cite{9322510,8741198}, spectral efficiency \cite{jung2021performance}, outage probability and asymptotic distribution of the sum-rate \cite{jung2019reliability}. 
However, the system performance of IRS-assisted two-way communications has not been studied thoroughly, and very limited number of results are available so far. 
In \cite{atapattu2020reconfigurable}, the outage probability and spectral efficiency of two-way communications assisted by IRS are investigated in reciprocal and non-reciprocal channels. 
In \cite{9324795}, fluctuating two-ray distribution in mmWave frequency is used to derive the outage probability and the average bit error probability for IRS aided systems.

Almost all of the aforementioned works are conducted with the assumption that all reflecting elements at IRS are able to reflect incident signals with the \textit{same constant amplitude} and the IRS  beamformer are designed to be \textit{optimally phase-matched} according to the channel state information (CSI).
However, \textit{the amplitude and phase of IRS elements can both be controlled independently} in practice via controlling over the resistance and capacitance of the integrated circuits in the IRS element, respectively \cite{guo2019weighted}.
Besides, as a complete passive equipment, IRS can not design the reflecting beamformer itself, appointing another node in the network to design the IRS reflecting beamfomer with collected CSI and then deliver the designed IRS reflecting beamformer to IRS may cause \textit{expensive signal overhead}, which makes such IRS beamforming design paradigm not feasible. Therefore, it is valuable to considered an IRS assisted communication system with given reflecting IRS beamformers, which are stored in the IRS node and obtained based on certain experiences, such as random phase rotation scheme \cite{9384311}.

In this letter, we focus on studying the outage probability of IRS-assisted full duplex two-way communication systems for any given reflecting coefficients of IRS. 
It is worth noting that deriving the outage probability with an given IRS beamformer, especially in the case with different amplitudes at different reflecting elements, is much more challenging than the case with an optimally designed phase-matched IRS beamformer.
To achieve this goal, we first transform the problem of deriving exact signal-to-noise ratio (SNR) expressions into a task of obtaining the distribution of inner product of two independent complex Gaussian random vectors.
Then, we develop the probability density distribution (PDF) and cumulative distribution function (CDF) of the inner product by applying the analytical framework that investigates the joint characteristic function of the real and imaginary parts of a complex variable.
Finally, closed-form outage probability expressions are obtained for IRS assisted two-way systems.
The correctness of the derived results are testified by numerical results which also illustrate the superiority of the considered IRS assisted two way systems comparing to the one way counterpart, where residual loop interference is considered.

\section{System Model}

Consider a two-way wireless communication system with two mobile users (namely, $U_1$ and $U_2$) and a single IRS as shown in Fig. 1. 
There is no direct link between mobile users due to serious fading and shadowing. 
The two users exchange messages by IRS with full-duplex communication strategy, which indicates two nodes transmitting and receiving information symbols simultaneously. 
Each user has two antennas which are responsible for signal transmission and reception respectively.
The pair of antennas are implemented far apart each other or constituted by non-reciprocal hardware.
Therefore, we assume that the forward and backward channels are non-reciprocal.
\begin{figure}[htbp]
	\vspace{-8 pt}
	\centerline{\includegraphics[width=7.2 cm]{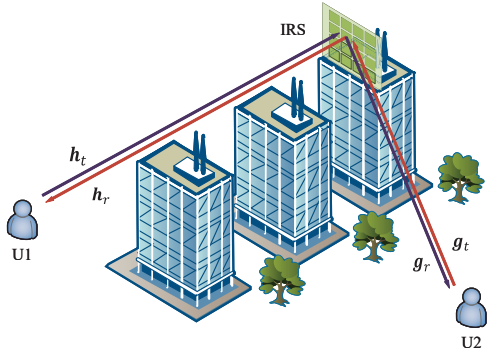}}
	\caption*{Fig. 1. The considered full duplex Two-way communication system assisted by an IRS with $N$ reflecting elements.}
	\label{fig1}
	\vspace{-8 pt}
\end{figure}

We assume that each reflecting element can continuously control both the amplitude and phase of the reflected signal. 
Most importantly, the phase shifts at IRS are arbitrarily given based on certain experiences, not just the simple case with phase-matched beamformer. 
They are denoted by a diagonal matrix $\boldsymbol{\Theta} = \operatorname{diag} \{\theta_1,...,\theta_N\} $ where $ \theta_i=|\theta_i| e^{j \phi_i }$, the amplitude $|\theta_i| \in \left[ 0,1\right]$, the phase shift $\phi_i \in \left[ 0,2 \pi \right) $, and $N$ denotes the number of reflecting elements. 
A lot of prior works have been conducted with the assumption that the amplitudes of all reflecting elements are \textit{constant}, which can only be seen as a special case of our letter. 

The channel coefficients of $U_1$-IRS link and $U_2$-IRS link are denoted as $\boldsymbol{h}_t=\left[ h_{t1},...,h_{tN} \right]^T$ and $\boldsymbol{g}_t=\left[ g_{t1},...,g_{tN} \right]^T$, respectively. 
Accordingly, the channels from IRS to $U_1$ and $U_2$ are $\boldsymbol{h}_r=\left[ h_{r1},...,h_{rN} \right]^T$ and $\boldsymbol{g}_r=\left[ g_{r1},...,g_{rN} \right]^T$, respectively. 
All the channels are assumed to be independent and identically distributed (i.i.d.) complex Gaussian fading, i.e., $\boldsymbol{h_t} \sim \mathcal{CN} (\boldsymbol{0},\sigma_{ht}^2 \boldsymbol{I}_N) $,
$\boldsymbol{h_r} \sim \mathcal{CN} (\boldsymbol{0},\sigma_{hr}^2 \boldsymbol{I}_N) $,
$\boldsymbol{g_t} \sim \mathcal{CN} (\boldsymbol{0},\sigma_{gt}^2 \boldsymbol{I}_N) $, $\boldsymbol{g_r} \sim \mathcal{CN} (\boldsymbol{0},\sigma_{gr}^2 \boldsymbol{I}_N) $. 
Furthermore, we denote the loop channels between the transmitting and receiving antenna of each user by $\tilde{\boldsymbol{h}}$ and $\tilde{\boldsymbol{g}}$, respectively.

At each time slot, $U_1$ and $U_2$ transmit their data to IRS, and IRS reflects the
received signal to $U_1$ and $U_2$. Therefore, the signal received by $U_1$ and $U_2$ are given by
\begin{subequations}
\begin{align}
y_1 {} & =  \underbrace{ \sqrt{P_{2}}\boldsymbol{h}_r^T \boldsymbol{\Theta g}_t s_{2} }_{\text{Desired-signal}} + \underbrace{ \sqrt{P_{1}}\boldsymbol{h}_r^T \boldsymbol{\Theta h}_t s_1 }_{\text{Self-interference}}  + \underbrace{ \sqrt{P_{1}} \tilde{\boldsymbol{h}} s_1 }_{\text{Loop-interference}} +n_1 , \\
y_2 {} & =  \underbrace{ \sqrt{P_{1}}\boldsymbol{g}_r^T \boldsymbol{\Theta h}_t s_{1} }_{\text{Desired-signal}} + \underbrace{ \sqrt{P_{2}}\boldsymbol{g}_r^T \boldsymbol{\Theta g}_t s_2 }_{\text{Self-interference}} + \underbrace{ \sqrt{P_{2}} \tilde{\boldsymbol{g}} s_2}_{\text{Loop-interference}} + n_2 .
\end{align}
\end{subequations}
where $P_{i},$ $i\in \{1,2\}  $ is transmitting power of $U_i$ and $s_i$ is the symbol that $U_i$ wants to transmit to the other user. 
Additive white Gaussian noise $n_1$ and $n_2$ are subject to $n_1 \sim \mathcal{CN} (0,\sigma_1^2 ) $, $n_2 \sim \mathcal{CN} (0,\sigma_2^2 ) $, respectively. 

Since the users have the global CSI, they can completely eliminate the self-interference.
To avoid loop interference, $U_1$ and $U_2$ apply some sophisticated loop interference cancellations, which results in residual interference \cite{atapattu2020reconfigurable}.
After interference cancellation, $U_1$ and $U_2$ receives
\begin{subequations}
\begin{align} \label{receivedsignal}
y_1^\prime = \sqrt{P_{2}}\boldsymbol{h}_r^T \boldsymbol{\Theta g}_t s_{2} + m_{1} + n_1 ,  \\
y_2^\prime = \sqrt{P_{1}}\boldsymbol{g}_r^T \boldsymbol{\Theta h}_t s_{1} + m_{2} + n_2 ,
\end{align}
\end{subequations}
where $m_{i},i\in \{1,2\}$ is the received residual loop interference of $U_i$ resulting from several stages of cancellation. We adopt the model where $m_{i}$ is subject to Gaussian distribution with zero-mean and $\sigma_{m_{i}}^2$ variance \cite{atapattu2020reconfigurable}. 
Further, the variance is characterized as $\sigma_{m_{i}}^2 = q_{i} P_i^{v_i}$ for $P_i\geq1$, where the two constrains, $q_{i}>0$ and $v_i\in \left[0,1 \right] $, depend on the cancellation method applied by users. 
Therefore, SNRs are expressed as 
\begin{subequations}
\begin{align} \label{SINR}
\gamma_1 = \frac{|\boldsymbol{h}_r^T \boldsymbol{\Theta g}_t|^2 P_{2}}{\sigma_1^2 + \sigma_{m_{1}}^2}  ,\\
\gamma_2 = \frac{|\boldsymbol{g}_r^T \boldsymbol{\Theta h}_t|^2 P_{1}}{\sigma_2^2 +\sigma_{m_{2}}^2} .
\end{align}
\end{subequations}

We define that the system outage occurs when at least one of the two users outages, i.e., $\gamma_1$ or $\gamma_2$ drops below their acceptable SNR threshold $\gamma_{t1}=2^{R_1}-1$ or $\gamma_{t2}=2^{R_2}-1$, respectively, where $R_1$ and $R_2$ are target rates of $U_1$ and $U_2$. Then the outage probability can be written as
\begin{equation}\label{TW_OP_defination}
P_{\text{out}} = Pr\{ \left\lbrace \gamma_1 <\gamma_{t1} \right\rbrace \cup \left\lbrace \gamma_2 <\gamma_{t2} \right\rbrace \}  .
\end{equation}

\section{Analysis of Outage probability}

In this section, we focus on outage probability of the system.
Firstly, we give the lemma as follows to obtain PDF of module of cascade channel $\boldsymbol{h}_r^T \boldsymbol{\Theta g}_t$ in (\ref{SINR}) with given $\boldsymbol{\Theta}$ as an example. 
Due to the equality of users, we assume that $\sigma_{ht}^2 = \sigma_{hr}^2=\sigma_{h}^2$, $\sigma_{gt}^2 = \sigma_{gr}^2=\sigma_{g}^2$ and $U_1$ and $U_2$ have the same target rate $R$.

\begin{lemma} \label{lemma1content}
Probability density distribution (PDF) of cascade channel $ z = \boldsymbol{h}_r^T \boldsymbol{\Theta g}_t  $ in IRS assisted two-way systems with given phase shift matrix $ \boldsymbol{\Theta} $ can be obtained as follows:
\begin{align}\label{lemma1}
p_z(r_z) = {} & r_z \sum_{i=1}^{N} C_i K_0 \left( \frac{2 r_z}{\sigma_{h} \sigma_{g} |\theta_i|} \right)\text{,}
\end{align}
where $r_z = |z|$, $K_v(\text{·})$ is the second kind modified Bessel function with the order v, 
$ C_i = \frac{a_i^{(N-2)}}{\prod_{j=1,j\neq i}^{N} (a_i - a_j)} $, $ a_i = \frac{1}{4} \sigma_{g}^2 \sigma_{h}^2 |\theta_i|^2, i\in \{1,...,N\} $.
\end{lemma}

\begin{proof}
	
To efficient present the proof, we construct the following new random variable
\begin{align}
	\boldsymbol{t}^{T}
	{} & = \boldsymbol{h}_r^T \boldsymbol{\Theta}\text{,}
\end{align}
where $t_i=h_{ri}\theta_i, i \in \left\{ 1,2,...,N \right\}$.
Since $\boldsymbol{h}$ satisfy the complex Gaussian distribution $\mathcal{CN} (\boldsymbol{0},\sigma_{h}^2 \boldsymbol{I}_N) $,
we obtain the distribution of $\boldsymbol{t}$ according to \cite{mallik2011distribution}, i.e., $\boldsymbol{t} \sim CN(\boldsymbol{0},\sigma_{h}^2 \tilde{\boldsymbol{\Theta})}$,
where $ \tilde{\boldsymbol{\Theta}} = diag[|\theta_1|^2,...,|\theta_N|^2] $.
Then we have
\begin{equation} \label{HGproduct}
	z = \boldsymbol{t}^T \boldsymbol{g}_t = \sum_{i=1}^{N} { h_{ri}\theta_i g_{ti} } \triangleq z_R + j z_I\text{,}
\end{equation}
where $ z_R  $ and $ z_I $ are the real and imaginary parts of $z$, respectively. According to (\ref{HGproduct}), it is clear that $z_R$ and $z_I$, conditioned on $\boldsymbol{t}$, are independent and their conditional distributions can be expressed as
\begin{subequations}
\begin{align} \label{ZRI}
	z_R|_{\boldsymbol{t}} \sim \mathcal{N}(\boldsymbol{0},\frac{\sigma_{g}^2}{2} \Vert \boldsymbol{t} \Vert^2)  \text{,}   \\
	z_I|_{\boldsymbol{t}} \sim \mathcal{N}(\boldsymbol{0},\frac{\sigma_{g}^2}{2} \Vert \boldsymbol{t} \Vert^2)  .
\end{align}
\end{subequations}
The conditional joint characteristic function of the conditioned $ z_R  $ and $ z_I $ \cite{mallik2011distribution} is expressed as
\begin{align} \label{conditionpdf}
	\Psi_{z_R,z_I | \boldsymbol{t}} (j\omega_1,j\omega_2 |\boldsymbol{t} ) 
	& = E\{\text{exp}\left[j(\omega_1 z_1 + \omega_2 z_2)| \boldsymbol{t}=\boldsymbol{t}_0\right]\}  \notag \\
	& = \text{exp}\left[-\frac{1}{4}(\omega_1^2 + \omega_2^2) \sigma_{g}^2 \parallel \boldsymbol{t}_0 \parallel^2 \right]  .
\end{align}
Since each element $t_i$ of random vector $\boldsymbol{t}$ is subject to complex Gaussian distribution, probability density function of $\boldsymbol{t}$ is
\begin{equation} \label{tpdf}
	f(\boldsymbol{t}) = \frac{1}{\pi^N \sigma_{hr}^{2N} \prod_{i=1}^{N}|\theta_i|^2} \text{exp}\left(-\sum_{i=1}^{N} \frac{|t_i|^2}{\sigma_{h}^2 |\theta_i|^2}\right).
\end{equation}
Then, combining (\ref{conditionpdf}) and (\ref{tpdf}) , the joint characteristic function of $z_R$ and $z_I$ is given by
\begin{align} \label{joint function}
	&\Psi_{z_R,z_I} (j\omega_1,j\omega_2) 
	= \int_{\boldsymbol{t} \in \boldsymbol{C}^N} \Psi_{z_R,z_I | \boldsymbol{t}} (j\omega_1,j\omega_2 | \boldsymbol{t}) f(\boldsymbol{t}) d \boldsymbol{t}   \notag \\
	& \stackrel{a}= \int_{\boldsymbol{t} \in \boldsymbol{C}^N} \frac{\text{exp}\left(-\sum_{i=1}^{N} \frac{|t_i|^2}{\sigma_{h}^2 |\theta_i|^2} -\frac{1}{4}(\omega_1^2 + \omega_2^2) \sigma_{g}^2 \parallel \boldsymbol{t} \parallel^2 \right)  }{\pi^N \sigma_{h}^{2N} \prod_{i=1}^{N}|\theta_i|^2}  d \boldsymbol{t}  \notag  \\
	&\stackrel{b}= \frac{\prod_{i=1}^{N} \int_{t \in C^N} \text{exp}\bigg\{ - \left[ \frac{1}{\sigma_{h}^2 |\theta_i|^2} + \frac{1}{4} (\omega_1^2 + \omega_2^2) \sigma_{g}^2 \right] |t_i|^2 \bigg\} dt_i}{\pi^N \sigma_{h}^{2N} \prod_{i=1}^{N}|\theta_i|^2}  \notag \\
	& \stackrel{c}= \frac{\prod_{i=1}^{N} \frac{\pi \sigma_{h}^2 |\theta_i|^2}{1 + \frac{1}{4} (\omega_1^2 + \omega_2^2) \sigma_{g}^2 \sigma_{h}^2 |\theta_i|^2}}{\pi^N \sigma_{h}^{2N} \prod_{i=1}^{N}|\theta_i|^2}  \notag \\
	&\stackrel{d}= \prod_{i=1}^{N} \left( 1 + \frac{1}{4}(\omega_1^2 + \omega_2^2) \sigma_{g}^2 \sigma_{h}^2 |\theta_i|^2 \right) ^{-1}  \text{,}
\end{align}
where (\ref{joint function}c) follows from  $\int_{-\infty}^{+\infty} \text{exp}(-ax^2) dx = \sqrt{\frac{\pi}{a}}, x\in\mathbb{R} $.

Supposing that PDF of $z$ is $p_z(x,y)$ where $x$ and $y$ are its real and imagine parts, after converting the joint characteristic function from Cartesian coordinates form to the Polar coordinate form $p_z(r_z,\beta_z)$, we have
\begin{align} \label{polar}
	&p_z(r_z) = \int_{0}^{2\pi} p_z(r_z,\beta_z) \,\mathrm{d} \beta_z \notag \\
	&= \frac{1}{4\pi^2} \int_{\beta_z=0}^{2\pi} r_z \int_{\omega_1 = -\infty}^{+ \infty} \int_{\omega_2 = -\infty}^{+ \infty} \Psi_{z_R,z_I} (j\omega_1,j\omega_2) \notag\\
	& \text{exp}\left[ -jr_z(\omega_1 \cos\beta_z + \omega_2 \sin\beta_z) \right] \,\mathrm{d} \omega_1 \,\mathrm{d} \omega_2 \,\mathrm{d} \beta_z  .
\end{align}
Substituting (\ref{joint function}) into (\ref{polar}), probability density distribution of cascade channel $ z $ is obtained in (\ref{pdf}) as shown on the top of the next page, where $I_\alpha(\text{·})$ is the first kind modified Bessel function and $J_\alpha(\text{·})$ is the first kind Bessel function of with the order $\alpha$. 
The equality (a) of (\ref{pdf}) is according to (3.339) in \cite{jeffrey2007table}, and equality (b) follows from (8.406) in \cite{jeffrey2007table}. 
\begin{figure*}[!t]
	\normalsize
	\setcounter{MYtempeqncnt}{\value{equation}}
	\vspace{-15pt}
	\begin{align}
	\label{pdf}
		p_z (r_z) = {} & \frac{r_z}{4\pi^2} \int_{0}^{\infty} z\left( \prod_{i=1}^{N} \frac{1}{1 + \frac{z^2}{4} \sigma_g^2 \sigma_h^2 |\theta_i|^2} \right) \int_{0}^{2\pi} \int_{0}^{2\pi} \text{exp}\left[ -j r_z z \cos(\beta_z - \phi)  \right] \,\mathrm{d} \beta_z \,\mathrm{d} \phi \,\mathrm{d} z  \notag \\
\stackrel{a} = {} &\frac{r_z}{4\pi^2} \int_{0}^{\infty} z\left( \prod_{i=1}^{N} \frac{1}{1 + \frac{z^2}{4} \sigma_g^2 \sigma_h^2 |\theta_i|^2} \right) \int_{0}^{2\pi} 2\pi I_0(-jr_zz)  \,\mathrm{d} \phi \,\mathrm{d} z  \notag \\
\stackrel{b}= {} & r_z \int_{0}^{\infty} z\left( \prod_{i=1}^{N} \frac{1}{1 + \frac{z^2}{4} \sigma_g^2 \sigma_h^2 |\theta_i|^2} \right) J_0(r_zz)  \,\mathrm{d} z
	\end{align}
	\setcounter{figure}{\value{MYtempeqncnt}}
	\hrulefill
	\vspace{-10pt}
\end{figure*}

Since it is arduous to compute the integral in (\ref{joint function}), we convert it from the product of a series of fractions to the sum of fractions as shown in
\begin{align} \label{prod_sum}
\prod_{i=1}^{N} \frac{1}{1+z^2 a_i} = {} & \sum_{i=1}^{N} \frac{C_i}{z^2 + \frac{1}{a_i}}  \text{,}
\end{align}
where $C_i,i\in \{1,2,...,N\}$ is the constant coefficient related to
$ a_i = \sigma_g^2 \sigma_h^2 |\theta_i|^2/4, i\in\{1,2,...,N\}$.
We can obtain each coefficient $C_i$ by multiple factorization:
\begin{equation} \label{C_i}
C_i = \frac{a_i^{(N-2)}}{\prod_{j=1,j\neq i}^{N} (a_i - a_j)} .
\end{equation}
Substituting (\ref{prod_sum}) and (\ref{C_i}) into (\ref{pdf}), we transform the integration to the sum of a series of Bessel function. 
Therefore, the the PDF of $z$ is given by
\begin{align} \label{pa_without_inte}
p_z(r_z) 
= {} & r_z \int_{0}^{\infty} z \left( \sum_{i=1}^{N} \frac{C_i}{z^2 + \frac{1}{a_i}} \right) J_0(r_zz) \,\mathrm{d} z \notag \\
\stackrel{a} = {} & r_z \sum_{i=1}^{N} C_i \left( \int_{0}^{\infty} \frac{z}{z^2 + \frac{1}{a_i}} J_0(r_zz) \,\mathrm{d} z \right) \notag \\
\stackrel{b} = {} & r_z \sum_{i=1}^{N} C_i K_0 \left( \frac{2 r_z}{\sigma_{h} \sigma_{g} |\theta_i|} \right) \text{,}
\end{align}
where the (\ref{pa_without_inte}b) step follows from Eq. (6.532.4) in \cite{jeffrey2007table}.
\end{proof}

Similarly, we can obtain the PDF of the cascade channel $z' = \boldsymbol{g}_r^T \boldsymbol{\Theta h}_t$ for $U_2$ as follows:
\begin{align} \label{lemma12}
p_z^\prime(r_z^\prime) = {} r_z^\prime \sum_{i=1}^{N} C_i K_0 \left( \frac{2 r_z^\prime}{\sigma_{g} \sigma_{h} |\theta_i|} \right) \text{,}
\end{align}
where $ C_i = \frac{a_i^{(N-2)}}{\prod_{j=1,j\neq i}^{N} (a_i - a_j)} $, $ a_i = \frac{1}{4} \sigma_{g}^2 \sigma_{h}^2 |\theta_i|^2, i\in \{1,...,N\} $.

As a consequence of Lemma 1, for the special case where the IRS model that each element has constant amplitude and continuous phase-shift, 
the joint character function in (\ref{joint function}) is simplified as 
\begin{align} \label{JF_constant_am}
\Psi_{z_R,z_I} (j\omega_1,j\omega_2) = \Biggl( 1 + \frac{1}{4}(\omega_1^2 + \omega_2^2) \sigma_{g}^2 \sigma_{h}^2  \Biggr)^{-N} .
\end{align}
Accordingly, the PDF of cascade channel yields
\begin{equation} \label{PDF_constant_am}
p_z(r_z) = \frac{4r_z^N}{\Gamma(N) \cdot (\sigma_{g} \sigma_{h})^{N+1}  }\cdot K_{N-1} \Bigl( \frac{2r_z}{\sigma_{g} \sigma_{h}} \Bigr)  .
\end{equation}

\begin{theorem}
If the SNR threshold of $U_1$ and $U_2$ are set as $\gamma_{{t}_1}$ and $\gamma_{{t}_2}$, respectively, the outage probability of IRS assisted two-way communication system is given by
\begin{equation}\label{theorem}
P_{\text{out} } = P_{out_1} + P_{out_2} - P_{out_1} P_{out_2},
\end{equation}
where for $j \in \{1,2\}$,
\begin{equation}
P_{out_j} = \sum_{i=1}^{N} \frac{C_i \left(\sigma_{h} \sigma_{g} |\theta_i|\right)^2}{2} \Biggl[ \frac{1}{2} -  \frac{ \sqrt{\gamma^{\prime}_{{t}_j}} K_1\left( \frac{2\sqrt{\gamma^{\prime}_{{t}_j}}}{\sigma_{h} \sigma_{g}|\theta_i|} \right)  }{\sigma_{h} \sigma_{g} |\theta_i|} \Biggr] ,
\end{equation}
\begin{equation}
\gamma^{\prime}_{{t}_j} = \left(  \frac{1}{\rho} + q_j ( \rho \sigma^2_{j} )^{v_j-1} \right)  \gamma_{{t}_j}
\end{equation}
and $ C_i $ and $a_i$ are given in Lemma \ref{lemma1content}.
\end{theorem}

\begin{proof}
To derive the outage probability of two-way system, the CDF of power of cascade channel $|\boldsymbol{h}_r^T \boldsymbol{\Theta g}_t|^2$ need to be obtained.
Firstly, based on PDF of $r_z$ obtained in Lemma \ref{lemma1content}, we can derive the PDF of $r_z^2$, which is denoted by $R_z$.
The PDF of $R_z$ can be obtained as follows:
\begin{equation}
p_{R_z} (R_z) = \sum_{i=1}^{N} \frac{C_i}{2}  K_0\left( \frac{2 (R_z)^{\frac{1}{2}}}{\sigma_{h} \sigma_{g} |\theta_i|} \right)  .
\end{equation}
Then, the CDF of $R_z$ is expressed as
\begin{align} \label{cdf}
&F_{R_z}(\xi) = {}  \int_{0}^{\xi} p_{R_z}(R_z) \,\mathrm{d} R_z \notag \\
&\stackrel{a} = {}  \int_{0}^{\xi} \sum_{i=1}^{N} \frac{C_i}{2} K_0\left( \frac{2 (R_z)^{\frac{1}{2}}}{\sigma_{h} \sigma_{g} |\theta_i|} \right) \,\mathrm{d} R_z \notag \\
&\stackrel{b}  = \sum_{i=1}^{N}  \frac{C_i \left(\sigma_{h} \sigma_{g} |\theta_i|\right)^2}{2}\left[ \frac{1}{2} -  \frac{\xi^{\frac{1}{2}}}{\sigma_{h} \sigma_{g} |\theta_i|}  K_1\left( \frac{2\xi^{\frac{1}{2}}}{\sigma_{h} \sigma_{g} |\theta_i|} \right) \right] \text{,}
\end{align}
where the equality (b) of (\ref{cdf}) follows from (12) in \cite{ding2020simpleNOMA} and  (6.561.8) in \cite{jeffrey2007table}.

Secondly, we assume that $\gamma^{\prime}_{{t}_j} = \frac{\gamma_{{t}_j} \left(  \sigma^2_{j} + \sigma^2_{m_j}\right) }{ P_j} $,  since $ \gamma_j = \frac{R_z P_j }{\sigma^2_{j} + \sigma^2_{m_j}} $, we have $P_{out_j} = Pr \left( \gamma_j < \gamma_{{t}_j} \right) = Pr \left(  R_z  < \gamma^{\prime}_{{t}_j} \right) $.
According to the transmit signal-to-noise ratio (SNR) of $U_j$, $ \rho_j = P_j / \sigma^2_{j} $ and loop interference residual $ \sigma^2_{m_j} = q_j P_j^{v_j}$, the equivalent SNR threshold of $U_j$ is expressed as $ \gamma^{\prime}_{{t}_j} = \left( \frac{1}{\rho_j} + q_j ( \rho_j \sigma^2_{j} )^{v_j-1}\right)  \gamma_{{t}_j} $. 
Therefore, the outage probability of $U_j$ is obtained as follows:
\begin{equation}\label{single_outage}
P_{out_j} = F_{R_z} \left(  \left(  \frac{1}{\rho_j} + q_j ( \rho_j \sigma^2_{j} )^{v_j-1} \right)  \gamma_{{t}_j} \right) ,
\end{equation}
where $\gamma_{{t}_j}$ is SNR threshold of $U_j$, $j=1,2$.

Finally, the outage probability of two-way system can be shown as
\begin{equation} \label{OP_threshold}
P_{\text{out}} = 1-(1-P_{out_1})(1-P_{out_2}) .
\end{equation}
Substituting (\ref{cdf}) and (\ref{single_outage}) into (\ref{OP_threshold}), the outage probability can be obtained as (\ref{theorem}) in the theorem.
\end{proof}

\begin{corollary}
For the special case that each element has constant amplitude and continuous phase-shift, i.e., $ \theta_i \in \bigl\{\theta_i \Big\arrowvert \theta_i = e^{j \phi_i}, \phi_i \in \left[0,2 \pi \right)  \bigr\} $, the outage probability of two-way system can be approximated as follows:
\begin{equation}\label{lemma2}
P_{\text{out}} \approx \frac{\gamma^{\prime}_{t_1}+\gamma^{\prime}_{t_2}-\gamma^{\prime}_{t_1}\gamma^{\prime}_{t_2}}{ N-1}  ,
\end{equation}
where $\gamma^{\prime}_{t_j}$ is equivalent SNR threshold of $U_j$.
\end{corollary}
\begin{proof}
Noting that when $N$ is large and amplitude of phase shift for all elements are constant, i.e. $ |\theta_i|=1$, the variable in (\ref{PDF_constant_am}) obeys the Nakagami-m distribution. The corresponding outage probability of $U_j,j \in \{1,2\}$ degrades into \cite{ding2020simpleNOMA}
\begin{equation} \label{OP_j_lemma2}
P_{\text{out}_j} = 1 - \frac{2\left(\gamma^{\prime}_{{t}_j} \right)^{\frac{N}{2}} }{\Gamma \left( N \right)} K_N \left( 2 \sqrt{ \gamma^{\prime}_{{t}_j} }\right)  .
\end{equation}
Since the Bessel function can be approximated \cite{jeffrey2007table} as follows:
\begin{equation}
K_n(z) \approx \frac{1}{2}\left( \frac{\left(n-1 \right)! }{\left( \frac{z}{2}\right)^n } - \frac{\left(n-2 \right)! }{\left( \frac{z}{2}\right)^{n-2} } \right)  \text{,}
\end{equation}
where $n \geq 2$ and $z\to 0$. Therefore, (\ref{OP_j_lemma2}) can be expressed as
\begin{equation}\label{one_out}
P_{\text{out}_j} \approx \frac{\gamma^{\prime}_{t_j}}{ N-1 }  .
\end{equation}
Combining (\ref{one_out}) and (\ref{TW_OP_defination}), the outage probability of the system is obtained in (\ref{lemma2}).
\end{proof}

Based on the above derivation, we conclude that with arbitrarily given $\boldsymbol{\Theta}$, the outage probability does not depend on the phase shift of each elements, but only on the amplitude of them.
The reason is that $\boldsymbol{\Theta}$ is arbitrarily given and independent on the channel coefficients.
Besides, the outage probability is a statistical measure of system reliability.
The conclusion is encouraging because it indicates that with only statistical CSI information, we can improve the reliability of the network by only adjusting the modulus of the reflection elements. 

\section{Simulation Results}
In this section, simulation results are presented to demonstrate the performance of IRS assisted full duplex two-way communication systems with $\sigma_{ht}^2=\sigma_{hr}^2=1$, $\sigma_{gt}^2=\sigma_{gr}^2=1$ and channel noise $\sigma_{1}^2=\sigma_{2}^2=1$. The variable amplitudes of each elements are set to $|\theta_i| = i/N, i \in \{ 1,2,...,N \}$ and the phase shifts are random. 

Fig. 2 plots the outage probability as a function of $\rho$ (we assume that  $\rho_1=\rho_2=\rho$) when loop interference $ \sigma^2_{m_{1}} = \sigma^2_{m_{2}} = \omega=10^{-4}$, i.e., $\sigma^2_{m_{1}}=\omega_{i}P_i^{v_i}$, where $v_i=0, \omega_{i}=10^{-4}$. 
It confirms the accuracy of the developed analytical results of CDF shown in Theorem 1. 
The presented illustrations include Monte Carlo simulating results with $10^6$ independent channel realizations for the outage probability.  
With the increasing of reflection elements quantity $N$, the outage probability decreases significantly.
\begin{figure}[htbp]
	\vspace{-20pt}
	\centerline{\includegraphics[width=8 cm]{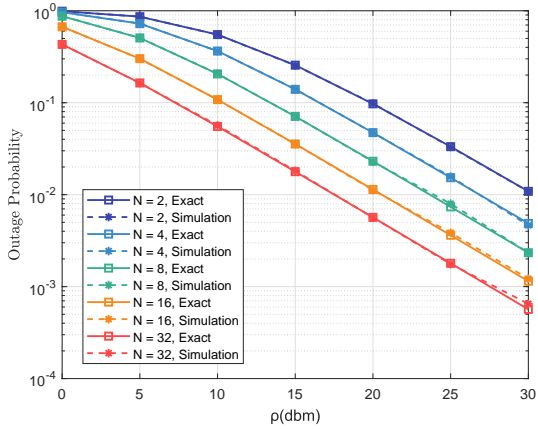}}                                                                                                                                                                                                                                                                                                                                                                                                                                                                                                                                                                                                                                                                                                                                                                                                                                                                                                                                                                                                                                                                                                                                                                                                                                                                                                                                                                                                                                                                                                                                                                                                                                                                                                                                                                                                                                                                                                                                                                                                                                                                                                                                                                                                                                                                                                                                                                                                                                                                                                                                                                                                                                                                                                                                                                                                                                                                                                                                                                                                                                                                                                                                                                                                                                                                                                                                                                                                                                                                                                                                                                                                                                                                                                                                                                                                                                                                                                                                                                                                                                                                                                                                                                                                                                                                                                                                                                                                                                                                                                                                                                                                                                                                                                                                                                                                                                                                                                                                                                                                                                                                                                                                                                                                                                                                                                                                                                                                                                                                                                                                                                                                                                                                                                                                                                                                                                                                                                                                                                                                                                                                                                                                                                                                                                                                                                                                                                                                                                                                                                                                                                                                                                                                                                                                                                                                                                                                                                                                                                                                                                                                                                                                                                                                                                                                                                                                                                                                                                                                                                                                                                                                                                                                                                                                                                                                                                                                                                                                                                                                                                                                                                                         
	\caption*{Fig. 2. Monte Carlo Simulation and Exact Results of Outage Probability.}
	\label{fig2}
	\vspace{-12pt}
\end{figure}

\begin{figure}[htbp]
	\vspace{-5pt}
	\centerline{\includegraphics[width=8 cm]{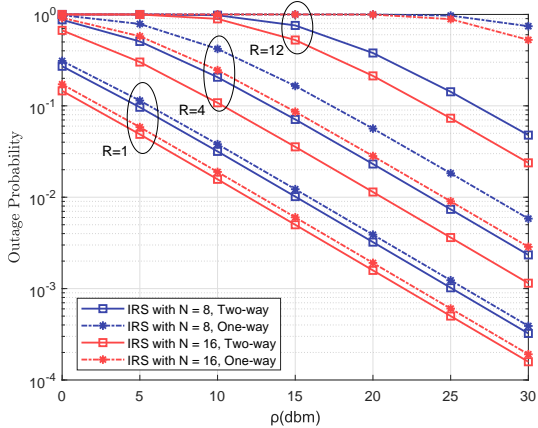}}
	\caption*{Fig. 3. Outage Probability of Two-way and One-way Communications.}
	\label{fig3}
	\vspace{-13 pt}
\end{figure}

Fig. 3 compares the outage performance of two-way and one-way communications, where the two-way target rate are 1,8,16 bits per channel use respectively and the loop interference $ \sigma^2_{m_{1}} = \sigma^2_{m_{2}} = \omega=10^{-4}$. 
We can observe that the outage performance for two-way communication is always better than that for one-way communication, which indicates the IRS assisted two-way scheme is more reliable than one-way scheme at the same effectiveness. 

Fig. 4 compares the outage probability for IRS of constant amplitude ($ |
\theta_i | = 1, i \in \{ 1,...,N \} $) and variable amplitude ($|\theta_i| = i/N, i \in \{ 1,...,N \}$) with different loop interference residual,  $\sigma^2_{m_{j}}=\omega_{j}$ and $\sigma^2_{m_{j}}=\omega_{j}P_j$, i.e., $v_j=0$ and $v_j=1$,  respectively, where $j=1,2$. It shows that outage probability for constant amplitude elements is lower than that with variable amplitude.   
Furthermore, the increase of the reflection element $N$ has a conspicuous effect on the improvement of reliability such that it is feasible to compensate the loss of loop interference residual by increasing $N$.
\begin{figure}[htbp]
	\vspace{-20pt}
	\centerline{\includegraphics[width=8.5 cm]{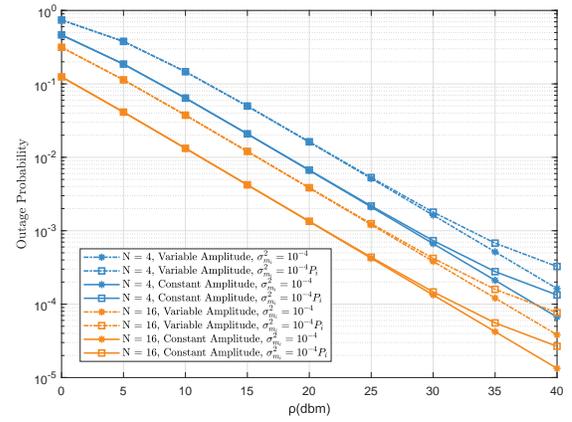}}
	\caption*{Fig. 4. Outage Probability of Two-way Communications with Constant Amplitude and Variable Amplitude.}
	\label{fig4}
\end{figure}

\section{Conclusion}
In this letter, we analyzed the PDF and CDF of cascade channel in IRS-assisted two-way communication systems with an arbitrarily given IRS beamformer.
The outage probability is derived in closed form based on reflecting elements assumptions with continuous amplitude and phase-shift.
The numerical results shows that the outage probability reduces significantly as the number of elements increases.
Based on the outage probability derived in this letter, the joint optimization of outage probability and other performance metrics such as sum-rate of two-way systems can to be further investigated.

\bibliographystyle{IEEEtran}
\bibliography{reference}

\end{document}